\documentclass[journal]{IEEEtran}

\usepackage[T1]{fontenc}
\usepackage[utf8]{inputenc}
\usepackage{amsmath,amssymb,amsfonts}
\usepackage{bm}
\usepackage{cite}
\usepackage{graphicx}
\usepackage{algorithm}
\usepackage{algorithmic}
\usepackage{booktabs}
\usepackage{array}
\usepackage{xcolor}
\usepackage[acronym]{glossaries}
\usepackage{subfig}
\usepackage{tabularx}

\newacronym{pai}{physical AI}{physical artificial intelligence}
\newacronym{ber}{BER}{bit error rate}
\newacronym{eier}{EIER}{event-inference error ratio}
\newacronym{twi}{TWI}{temporal window of integration}
\newacronym{eir}{EIR}{event-inference reliability}
\newacronym{snr}{SNR}{signal to noise ratio}
\newacronym{sinr}{SINR}{signal to interference plus noise ratio}
\newacronym{cdf}{CDF}{cumulative distribution function}
\newacronym{bs}{BS}{base-station}
\newacronym{aoi}{AoI}{age-of-information}
\newacronym{mi}{MI}{mutual-information}
\newacronym{map}{MAP}{maximum a posteriori}
\newacronym{sr}{SR}{scheduling request}
\newacronym{pt}{PT}{packet transmission}
\newacronym{pmf}{PMF}{probability mass function}

\newtheorem{definition}{Definition}

\newtheorem{proposition}{Proposition}

\newtheorem{remark}{Remark}

\newcommand{\E}{\mathbb{E}}
\newcommand{\Prb}{\mathbb{P}}
\newcommand{\cE}{\mathcal{E}}
\newcommand{\cO}{\mathcal{O}}
\newcommand{\cC}{\mathcal{C}}

\newcommand{\bO}{\bm{O}}

\title{Event-Inference Reliability for Physical AI over Wireless Networks}

\author{Anup~Mishra,~\IEEEmembership{Member,~IEEE} and Petar~Popovski,~\IEEEmembership{Fellow,~IEEE}\vspace{-0.8cm}

\thanks{A. Mishra and P. Popovski are with the Department of Electronic Systems, Aalborg University, Denmark (\textit{email}: anmi@es.aau.dk, petarp@es.aau.dk).}
}

\begin{document}
\maketitle

\begin{abstract}
Wireless-enabled \gls{pai} systems call for a shift from reliable data delivery to reliable inference of physical events. The relevant question is not only whether packets arrive, but whether the set of cues available at the decision node, i.e., the evidence, is sufficiently timely and informative to support reliable inference about the event. Accordingly, this paper develops a framework in which \emph{\gls{eir}} is determined jointly by cue informativeness, cue availability, and temporal admissibility. The latter is determined by the downstream task requirement and represented through the \emph{usefulness horizon}. We define \gls{eier} as the normalised residual event uncertainty after incorporating admitted cues, and \gls{eir} as the corresponding normalised uncertainty reduction, both conditioned on decision-node context. We further distinguish the evidence-limited Bayes benchmark from operational performance of a particular inference engine and derive an entropy-based lower bound on the minimum achievable event error from the same decision-node information. The framework then enables an event-aware wireless design interface for cue prioritisation, cue-reliability allocation, and event-inference coverage characterisation. A multiclass indoor activity-inference study combining empirical cue likelihoods with wireless delivery instantiates the framework and demonstrates how it characterises \gls{eir} under finite usefulness horizons.
\end{abstract}

\begin{IEEEkeywords}
Physical AI, multimodal sensing, event perception, temporal window of integration, information theory.
\end{IEEEkeywords}

\glsresetall

\section{Introduction}\label{sec:introduction}
\Gls{pai} refers to intelligent systems whose perception, reasoning, and
action are coupled to the physical world in real time~\cite{miriyev2020skills,liu2024embodied,mishra2025temporal}.
Examples include autonomous vehicles, industrial automation systems, smart
infrastructure, embodied assistants, and networked sensing platforms. In such systems, intelligence depends on the ability to infer physical states or events from noisy, incomplete, and heterogeneous sensory cues that collectively constitute the evidence available for inference, and to use this inference to support timely action~\cite{Mondal2025,popovski2022perspective,hashash2026active}.
Recent active-inference views of \gls{pai}, alongside broader multimodal intelligence architectures, emphasize this belief-updating role of sensory evidence~\cite{de2026active,hashash2026active}.
In these views, an agent uses an internal model to form expectations,
interpret new observations, update its beliefs when observations deviate
from those expectations, and adapt its actions accordingly~\cite{friston2010free,friston2017active}.
This perspective explains why sensory evidence is central to physical intelligence: without timely and reliable sensory cues, the agent cannot form reliable beliefs about the underlying physical state.
\par In wireless-enabled \gls{pai}, however, such sensory cues are not
automatically available to the agent or edge decision node. They are
generated by distributed sensors, transformed through local computation or
compression, transported over wireless links, aligned with other cues,
and then used for remote inference under stringent timing
constraints~\cite{popovski2022perspective,strinati2021beyond,thomas2026compositional}.
This makes the wireless system part of the \textit{evidence interface} rather than
merely a data pipe: it determines which sensory cues reach the decision node, when they arrive, and whether they arrive in time to remain useful for event inference~\cite{mishra2025temporal}. Fig.~\ref{fig:eir_toy_example} illustrates this distinction using three
generic cues, which can be interpreted, for example, as wearable,
object-attached, and ambient cues in the indoor activity-inference setting
instantiated later in Section~\ref{sec:operationalisation_activity_event}.
Cue availability and informativeness need not align: a readily available
cue may be weakly informative, while a highly informative cue may arrive
too late to contribute to inference, showing how wireless delivery
temporally shapes the evidence.
\begin{figure}[t]
\centering
\includegraphics[width=\columnwidth]{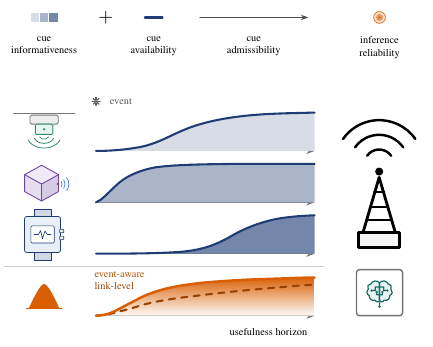}
\caption{\small Illustration of how cue informativeness, availability, and
wireless resource allocation jointly shape reliable inference.\vspace{-0.7cm}}
\label{fig:eir_toy_example}
\end{figure}

\par This evidence-interface view has complementary biological grounding
on both sides of the perception-action loop. On the inference side,
active-inference views of \gls{pai} draw on human cognition to explain how
internal models, prediction errors, and belief updates support
action~\cite{hashash2026active}. On the evidence-admission side, human
multisensory perception suggests that sensory inputs are useful only when
they are \textit{spatially congruent}, \textit{temporally compatible}, and
\textit{reliable} as evidence of a common physical
cause~\cite{kording2007causal,SPENCE2003R519}. These properties determine, respectively, whether cues relate to the same
event, fall within a plausible \gls{twi}, and how much influence
each cue should have on the resulting percept, with more reliable cues
weighted more strongly~\cite{ernst2002humans}.
\par These perceptual principles motivate a wireless admissibility
view in which received data are not automatically treated as event
evidence. The \gls{twi} framework proposed in~\cite{mishra2025temporal}
organises the timing dimension around three event-relative primitives:
\textit{simultaneity}, \textit{causality}, and \textit{usefulness}. {For each event-inference instance, an event reference time $t$ provides the
temporal anchor for these relations and for the corresponding usefulness
horizon. The resulting \gls{twi} conditions provide generic system-level
criteria for temporal admissibility across different sensing and inference
services: wireless delivery determines which cues become available, while
the \gls{twi} determines which of those cues are admissible for the current
inference instance. As illustrated in Fig.~\ref{fig:eir_toy_example},
however, temporal admissibility alone does not determine how much the
admitted evidence resolves uncertainty about the particular event being
inferred.}
\par {This paper therefore develops an \gls{eir} framework that operates on
the admitted evidence, complementing the system-level temporal admissibility
provided by the \gls{twi} by characterising reliability for a specific
event-inference task through the event prior and event-conditioned cue
statistics.} \textit{We regard event inference as reliable when the admitted
evidence, together with decision-node context, distinguishes the underlying
event from alternative event states with sufficiently low residual
uncertainty.} Accordingly, we define \gls{eier} as the fraction of the
context-conditioned event uncertainty that remains after cue admission, and
\gls{eir} as the complementary fraction resolved by the admitted cue set.
This event-level characterisation couples cue informativeness, availability,
and temporal admissibility, providing a basis for event-aware wireless
design.
\subsection{Related Work}
\label{subsec:related_work_positioning}
\par Active-inference and free-energy formulations provide an agent-level
account of perception and action through belief updating under an internal
generative model, often using variational free-energy or expected
free-energy minimisation~\cite{friston2010free,friston2017active}. Recent
\gls{pai}-oriented works extend this view to engineered physical agents and
world-model-based test-time adaptation, where sensory observations support
prediction-error detection, reasoning, and policy
adaptation~\cite{de2026active,hashash2026active}. These works establish the importance of sensory evidence for physical
intelligence, but generally abstract from the network-mediated conditions
under which distributed sensory cues become available, remain temporally
admissible, and provide sufficient information for belief updating, leaving
open how reliability should be defined at this evidence interface.
\par Wireless reliability has traditionally been characterised through
link-level metrics such as rate, outage, latency, and \gls{aoi}, reflecting
a communication paradigm centred on reliable symbol or bit
transmission~\cite{strinati2021beyond,popovski2022perspective}. Semantic and
goal-oriented communication extends this link-level focus toward the meaning,
relevance, and effectiveness of delivered information, including
semantic-aware sensing~\cite{zhu2025sansee}, task-specific semantic
protocols~\cite{seo2023semantic}, and compositional semantic communication
for heterogeneous \gls{pai} devices~\cite{thomas2026compositional}.
Complementing this shift, \cite{popovski2022perspective}
posits that application relevance also depends on timing across information
generation, processing, delivery, and reconstruction.
\par This shift toward task-relevant information is particularly important in multimodal sensing, where heterogeneous sensory cues of the same physical
process may arise from time series, images, audio, video, location traces,
and contextual records~\cite{Mondal2025}. Distributed detection and
data-fusion methods have studied the combination of local sensor reports
under noisy observations and heterogeneous sensor
reliabilities~\cite{varshney1997distributed,viswanathan1997distributed},
including spatial correlation and location-dependent
reliability~\cite{OuldAhmedVall2012}. These works support interpreting
sensor reports as evidence about an underlying physical process rather than
merely as independently delivered packets.
\par The temporal structure of this evidence must also be preserved.
Beyond link-level latency and \gls{aoi}, recent work treats timing as an
event-centric admissibility problem. The work in~\cite{Simul_Causal}
identifies chronology, simultaneity, causality, and timestamping as
challenges for perceptive wireless systems. The work in~\cite{Joao_SV}
uses \gls{twi}-based processing to preserve simultaneity under propagation,
computation, and communication delays, while the \gls{twi}
framework in~\cite{mishra2025temporal} organises timing around
\textit{simultaneity}, \textit{causality}, and \textit{usefulness}.
Together, these works shift timing from packet-level freshness toward the
temporal admissibility of evidence about physical events.
\par Existing work therefore addresses complementary layers of the evidence interface: wireless and \gls{twi}-oriented approaches determine which cues become available and remain temporally admissible, while inference and fusion approaches operate on the resulting evidence. What remains missing is a reliability measure that connects these layers at the event level, with direct implications for wireless design. The resulting design distinction is depicted in Fig.~\ref{fig:eir_toy_example}: link-level control allocates cue reliability based on delivery requirements, whereas event-aware control additionally accounts for the inferential contribution of each cue, thereby enabling more reliable inference under the same temporal constraints. To this end, an information-theoretic formulation provides the missing measure by quantifying how much the context and admitted cue set resolve uncertainty about the underlying event~\cite{gamal2026information}, motivating the event-decoding framework developed in this paper.

\subsection{Problem Statement and Contributions}
\label{subsec:problem_statement_contributions}
The discussion above leads to the following question:
\textit{given the decision-node context and the cue set admitted as per 
\gls{twi} constraints, how much uncertainty about the underlying event remains, and what
does this imply for achievable event-decision performance?}
Answering this question requires an event-level reliability formulation that
maps the admitted evidence to residual event uncertainty and event-error
performance while preserving the dependence of that evidence on wireless
delivery and control. Complementing the three-cue illustration in
Fig.~\ref{fig:eir_toy_example}, Fig.~\ref{fig:eir_inference_layer} situates
this reliability layer within the broader perception--action loop.

\begin{figure*}[t]
    \centering
    \includegraphics[width=\textwidth]{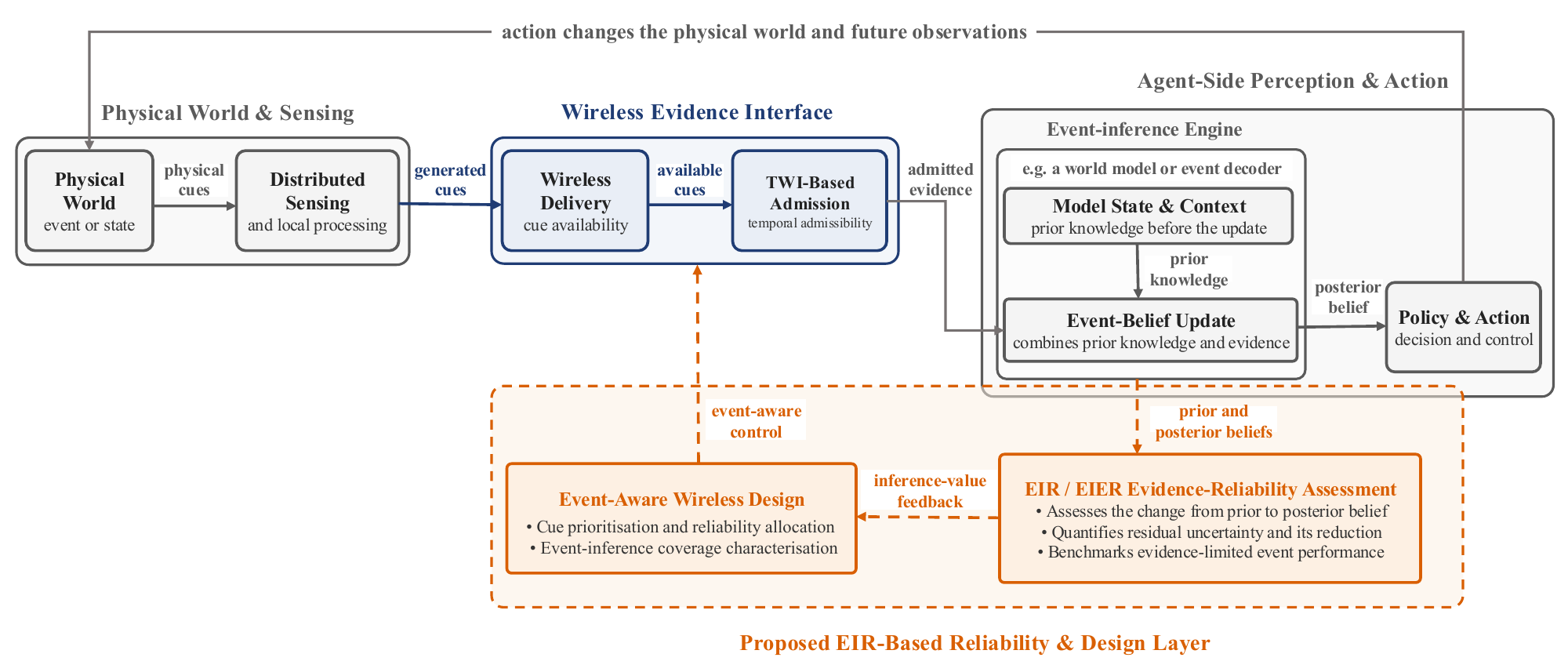}
    \caption{\small Wireless-enabled event inference within a closed
perception--action loop. Sensing, wireless delivery, and TWI-based
admission feed an event-belief update whose posterior drives policy and action, closing the outer loop; EIER/EIR assess the update, driving event-aware wireless design and closing the inner loop.\vspace{-0.5cm}}
    \label{fig:eir_inference_layer}
\end{figure*}
\par The main contributions of this paper are as follows.
\begin{itemize}
\item We formulate reliability in wireless-enabled \gls{pai} as a
context-aware event-decoding problem, in which a decision node infers an underlying physical event from its local context and the multimodal cue set
admitted within the relevant \gls{twi}. The formulation separates cue informativeness, characterised by the event-conditioned observation law; cue availability, determined by whether and when a generated cue reaches the decision node under sensing, computation, and communication delays; and temporal admissibility, determined here by whether the available cue satisfies the \gls{twi} usefulness condition.

\item We define \gls{eier} as the normalised residual uncertainty of the physical event conditioned on decision-node context and admitted evidence, and \gls{eir} as the corresponding normalised uncertainty reduction. We further distinguish the minimum event error supported by the available
evidence from additional performance loss associated with an implemented
inference engine, when such a decoder-specific gap is present.

\item We derive an entropy-based lower bound that relates residual event
uncertainty to the minimum achievable event error, providing an
information-theoretic constraint on evidence-limited event-decision
performance. Building on this connection, we develop an event-aware wireless design interface in which cue-reliability allocation accounts for the inferential contribution of different cues rather than delivery reliability alone.
\item We instantiate the framework in a \gls{twi}-gated multiclass
activity-inference task combining empirical cue likelihoods with stochastic
wireless delivery. The study evaluates the event posterior, \gls{eier},
\gls{eir}, and Bayes event error under event-relative timing constraints,
and shows that event-aware cue-reliability allocation prioritises
informative cues and attains matched event-inference coverage at shorter
usefulness horizons than link-level allocation.
\end{itemize}
\par \textit{Notation:} Random variables are denoted by uppercase letters and their realizations by lowercase letters when needed. Bold symbols denote vectors, and calligraphic symbols denote sets or alphabets. The probability and expectation operators are denoted by $\Prb(\cdot)$ and $\E[\cdot]$, respectively. Entropy and mutual information are denoted by $H(\cdot)$ and $I(\cdot;\cdot)$, with conditional forms written in the usual way. Unless otherwise stated, logarithms are taken to base $2$, so all information quantities are measured in bits.

\par \textit{Organization:} Section~\ref{sec:system_model} introduces the context-aware event-inference system model and the decision-node information available for reliability evaluation. Section~\ref{sec:eier_and_reliability} defines \gls{eier}, \gls{eir}, and the corresponding evidence-limited and decoder-specific event-error measures. Section~\ref{sec:it_bound_event_error} relates residual event uncertainty to event-error limits and develops the event-aware wireless design viewpoint. Section~\ref{sec:operationalisation_activity_event} instantiates the framework for \gls{twi}-gated multimodal activity inference under wireless cue delivery and presents the numerical evaluation. Section~\ref{sec:conclusion} concludes the paper.

\section{System Model}
\label{sec:system_model}
This section formalises the context-aware event-inference model used
throughout the paper. We consider a decision node located at the \gls{bs}
that infers a task-relevant physical event from prior contextual knowledge
and multimodal cues generated by a distributed sensing
process~\cite{Mondal2025,OuldAhmedVall2012}. Cue delivery and the associated
end-to-end timing determine which generated cues become available and
temporally admissible under the relevant \gls{twi}
constraints~\cite{mishra2025temporal,Joao_SV}. The resulting admitted cue
set constitutes the evidence used for event inference and reliability
evaluation.

\subsection{Physical Event, Decision-Node Context, and Prior}
\label{subsec:physical_event_context_prior}
Each event-inference instance is associated with an event reference time
$t$, which provides the temporal anchor for the event-belief update.
Depending on the application, $t$ may correspond to the occurrence or
onset of a task-relevant event, an estimate thereof when that instant is
not directly observable, or an application-defined reference time for a
persistent physical state or recurrent inference task. Relative to $t$,
the usefulness horizon $\Delta$ specifies the interval within which
arriving cues satisfy the usefulness condition considered here, and the
event-belief update is formed by $t+\Delta$ from the resulting admitted
evidence~\cite{mishra2025temporal}.
The value of $\Delta$ is task-dependent; for example, in high-mobility
settings, delayed evidence may rapidly lose its relevance to the downstream
task. Subsequently, let
\begin{equation}
    E_t \in \cE = \{1,2,\ldots,K\}
    \label{eq:event_source}
\end{equation}
denote the physical event or state of interest associated with reference time
$t$, where $\cE$ is a finite event alphabet and $K=|\cE|$ is the number of
event classes. The alphabet is selected at the semantic resolution required
by the downstream task; for example, $E_t$ may represent a pedestrian state
in autonomous mobility, a machine condition in industrial monitoring, or an
environmental change in distributed sensing~\cite{Mondal2025,OuldAhmedVall2012}.
\par Let $C_t\in\cC$ denote the context available at the decision node
before the current cue set is incorporated. Depending on the application,
$C_t$ may include the previous event belief, a latent world-model state,
recent event history, geometry, location, task information, or the current
operating mode~\cite{liu2024embodied}. Importantly, $C_t$ excludes the
current cues whose incremental contribution is evaluated in the present
event-belief update.
\par For a realised context $C_t=c_t$, the decision node forms the
context-conditioned prior
\begin{equation}
    p_t(e\mid c_t)
    =
    \Prb(E_t=e\mid C_t=c_t),
    \quad e\in\cE ,
    \label{eq:context_prior}
\end{equation}
which represents its belief about the event before incorporating the newly
admitted cue set. The corresponding prior uncertainty for the realised
context is
\begin{equation}
    H_t(E_t\mid c_t)
    =
    -\sum_{e\in\cE}
    p_t(e\mid c_t)\log_2 p_t(e\mid c_t).
    \label{eq:realised_context_entropy}
\end{equation}
Averaging over the context distribution gives the conditional event
entropy~\cite{CoverThomas2006}
\begin{equation}
\begin{split}
    H(E_t\mid C_t)
    &=
    \E_{C_t}
    \left[
        H_t(E_t\mid C_t)
    \right]\\
    &=
    -\E_{C_t}
    \left[
        \sum_{e\in\cE}
        p_t(e\mid C_t)\log_2 p_t(e\mid C_t)
    \right].
    \label{eq:context_entropy}
\end{split}
\end{equation}
Thus, $H_t(E_t\mid c_t)$ quantifies event uncertainty for the realised
context, whereas $H(E_t\mid C_t)$ is the corresponding average across
contexts; both are measured before incorporating the current cue set, i.e., evidence.
\begin{remark}
When no explicit context is modelled, $C_t$ may be treated as constant.
The context-conditioned prior and entropy then reduce to
$p_t(e)=\Prb(E_t=e)$ and $H(E_t)$, respectively.
\end{remark}

\subsection{Multimodal Cue Generation and Informativeness}
\label{subsec:multimodal_sensing_cue_informativeness}
For the physical event or state $E_t$, the distributed sensing process
generates modality-specific cues~\cite{Mondal2025,varshney1997distributed}.
Let
\begin{equation}
    O_t^{(m)} \in \cO_m,
    \quad
    m\in\{1,\ldots,M\},
    \label{eq:modality_observation}
\end{equation}
denote the cue generated by modality $m$, where $\cO_m$ is its observation
space. The corresponding multimodal cue vector is given by
\begin{equation}
    \bO_t
    =
    \left(
        O_t^{(1)},O_t^{(2)},\ldots,O_t^{(M)}
    \right).
    \label{eq:multimodal_observation}
\end{equation}
Its event-dependent statistics are described by
\begin{equation}
    p(\bO_t\mid E_t),
    \label{eq:observation_channel}
\end{equation}
which captures sensing noise, partial observability, modality-specific
degradation, spatial configuration, and statistical dependence across
modalities~\cite{viswanathan1997distributed,OuldAhmedVall2012}. Additional
sensing-side variables may be included as conditioning variables when they
affect cue generation; $C_t$, however, denotes the decision-node context
used to form the event prior before the current cues are incorporated.
\par The observation law in \eqref{eq:observation_channel} determines cue
informativeness. A cue is informative to the extent that its
event-conditioned distributions distinguish among the event classes and
therefore reduce uncertainty about $E_t$. Conversely, if different classes
induce nearly indistinguishable cue distributions, the cue provides little
discriminative information regardless of whether it is successfully
delivered and admitted. The joint observation law further captures
statistical dependence across modalities and hence whether their information
about the event is redundant or complementary~\cite{kording2007causal}.
\par Note that informativeness is a sensing-side property of the cue statistics,
whereas wireless delivery and \gls{twi}-based admission determine whether
that information becomes part of the evidence available at the decision
node. The next subsection formalises this mapping from the sensing-side
vector $\bO_t$ to the admitted evidence used for event inference.

\subsection{\Gls{twi} Admissibility of Delivered Cues}
\label{subsec:twi_admissibility_delivered_cues}

Cues generated by the sensing process do not automatically constitute
usable evidence for the current event-belief update. Following the
\gls{twi} framework in~\cite{mishra2025temporal}, temporal admissibility is
governed by simultaneity, causality, and usefulness. In the present model,
we focus explicitly on usefulness through the horizon $\Delta$, while
simultaneity and causality are assumed to be enforced upstream. These
conditions can alternatively be incorporated as additional admission
constraints without changing the structure of the event-level reliability
framework.
\par For modality $m$, let $\Gamma_m$ denote the end-to-end delivery path
associated with cue $O_t^{(m)}$, and let $T_{\Gamma_m}$ denote its
event-relative delay~\cite{Simul_Causal,mishra2025temporal}.
Accordingly, $T_{\Gamma_m}$ is the time from the event reference time $t$ to the
availability of cue $m$ at the decision node and may include event
propagation, sensing, local computation or compression, and wireless
communication delays.
Cue $m$ can contribute to the event-belief update only if it is
successfully delivered and satisfies
\begin{equation}
    T_{\Gamma_m}\leq\Delta .
    \label{eq:usefulness_condition}
\end{equation}
The usefulness horizon $\Delta$ is defined at the event-inference level,
according to the temporal validity of the multimodal cue set for the
downstream task, rather than as a link-level communication requirement.
Once this event-level $\Delta$  is specified, it can be translated into
packet-level delivery deadlines for the individual cues, accounting for the
time consumed by sensing, computation, propagation, and other components of
their end-to-end paths. Conventional deadline-based wireless mechanisms can
then be used to enforce timely delivery.

\par For each modality $m$, we define the probability that the cue is both
successfully delivered and available within the usefulness horizon as
\begin{equation}
    r_m(\Delta)
    =
    \left(1-p_m^{\mathrm{drop}}\right)
    F_{\Gamma_m}
    \left(
        \Delta\mid\mathrm{no~drop}
    \right),
    \label{eq:path_success_timeliness_probability}
\end{equation}
where $p_m^{\mathrm{drop}}$ is the probability of delivery failure and
\begin{equation}
    F_{\Gamma_m}
    \left(
        \Delta\mid\mathrm{no~drop}
    \right)
    =
    \Prb
    \left(
        T_{\Gamma_m}\leq\Delta
        \mid\mathrm{no~drop}
    \right)
    \label{eq:conditional_delay_cdf}
\end{equation}
is the conditional \gls{cdf} of the event-relative delay. The components of
$T_{\Gamma_m}$ can be modelled using the path formulation
in~\cite[eqs.~(2)--(16)]{mishra2025temporal}.

\begin{remark}[Reference-time implementation]
When the event reference time $t$ corresponds to a physical event occurrence or onset, that instant need not be directly observable. Suppose that
a runtime estimate $\widehat{t}$ satisfies
$0\leq\widehat{t}-t\leq\delta_0$ for a known conservative margin
$\delta_0$. Enforcing
\begin{equation}
    t_m^{\mathrm{arr}}-\widehat{t}
    \leq
    \Delta-\delta_0
\end{equation}
then guarantees
$t_m^{\mathrm{arr}}-t\leq\Delta$. Thus, the usefulness condition can be
enforced conservatively even when the event reference time is not directly observed.
\end{remark}
\par Hence, $r_m(\Delta)$ gives the probability that cue $m$ passes the
delivery-and-usefulness gate and can enter the admitted cue set used for
event inference.
\subsection{Decision-Node Evidence for Event Inference}
\label{subsec:decision_node_evidence_event_inference}
The delivery and \gls{twi} conditions above determine which generated cues
enter the current event-belief update. For each modality $m$, let
\begin{equation}
    A_{t,m}^{(\Delta)}\in\{0,1\}
    \label{eq:cue_admission_indicator}
\end{equation}
denote the admission indicator, where $A_{t,m}^{(\Delta)}=1$ if cue $m$
is successfully delivered and temporally admissible by $t+\Delta$, and
$A_{t,m}^{(\Delta)}=0$ otherwise. The corresponding decision-node cue is
represented as
\begin{equation}
    \widetilde{O}_t^{(m,\Delta)}
    =
    \begin{cases}
        O_t^{(m)}, & A_{t,m}^{(\Delta)}=1,\\
        \bot, & A_{t,m}^{(\Delta)}=0,
    \end{cases}
    \label{eq:admitted_modality_observation}
\end{equation}
where $\bot$ denotes the absence of an admitted cue. The admitted evidence
is therefore
\begin{equation}
    \widetilde{\bO}_t^{(\Delta)}
    =
    \left(
        \widetilde{O}_t^{(1,\Delta)},
        \ldots,
        \widetilde{O}_t^{(M,\Delta)}
    \right).
    \label{eq:decision_node_evidence}
\end{equation}
If cue-specific metadata, such as timestamps, confidence scores, or
uncertainty estimates, are used by the inference engine, they will be
included as part of the corresponding decision-node cue representation.
\par The information available for the event-belief update is thus
\begin{equation}
    \left(
        C_t,\widetilde{\bO}_t^{(\Delta)}
    \right).
    \label{eq:decision_node_information_pair}
\end{equation}
The posterior beliefs and event-reliability quantities developed
subsequently in this paper are defined with respect to this decision-node information.

\subsection{Event-Belief Update and Decoding}
\label{subsec:event_decoder_posterior_inference}

The perception stack at the \gls{bs} combines the decision-node context
with the admitted evidence; see Fig.~\ref{fig:eir_inference_layer}. For
each $e\in\cE$, the underlying joint law induces the posterior
\begin{equation}
    p
    \left(
        e\mid C_t,\widetilde{\bO}_t^{(\Delta)}
    \right),
    \label{eq:true_posterior}
\end{equation}
which represents the event belief supported by the available decision-node
information. An implemented inference engine produces
\begin{equation}
    q_\theta
    \left(
        e\mid C_t,\widetilde{\bO}_t^{(\Delta)}
    \right),
    \label{eq:posterior_approx}
\end{equation}
where $q_\theta$ coincides with the posterior above when it is available
and exploited exactly, and otherwise represents the belief produced by an
approximate Bayesian, learned perception, or world-model-based inference
mechanism. Dependence on $\Delta$ arises through
$\widetilde{\bO}_t^{(\Delta)}$, since different usefulness horizons can
admit different cue sets.

\par The corresponding hard event estimate is obtained using the
\gls{map} rule~\cite{varshney1997distributed,viswanathan1997distributed}:
\begin{equation}
    \hat{E}_t^{(\Delta)}
    =
    \arg\max_{e\in\cE}
    q_\theta
    \left(
        e\mid C_t,\widetilde{\bO}_t^{(\Delta)}
    \right).
    \label{eq:map_event_decoder}
\end{equation}

\par This formulation separates a property of the available information
from the performance of the inference engine used to exploit it. The joint
law of
$\left(E_t,C_t,\widetilde{\bO}_t^{(\Delta)}\right)$ determines the
posterior in \eqref{eq:true_posterior} and hence the uncertainty and minimum
event error supported by the admitted evidence. The implemented posterior
$q_\theta$ determines the operational event decision. When the underlying
posterior is available and used exactly, $q_\theta=p$ and this distinction
collapses to the Bayes-optimal case.

\begin{remark}[Relation to active inference]
Within an active-inference interpretation, $E_t$ represents the physical
state or event to be inferred, $C_t$ determines the predictive prior, and
$\widetilde{\bO}_t^{(\Delta)}$ provides the network-admitted sensory
evidence used for the posterior belief update.
\end{remark}

\section{\Gls{eier} and \gls{eir}}
\label{sec:eier_and_reliability}

The system model identifies
$\left(C_t,\widetilde{\bO}_t^{(\Delta)}\right)$ as the information available
for the event-belief update. We now quantify how much of the
context-conditioned event uncertainty remains after incorporating the
admitted evidence. Specifically, \gls{eier} measures the unresolved
fraction of this uncertainty, while its complement, \gls{eir}, measures the
fraction resolved by the admitted evidence. We then relate these quantities
to evidence-limited event error and distinguish it from the error achieved
by the implemented inference engine.

\par All quantities in this section are defined with respect to the joint
distribution of
\begin{equation}
    \left(
        E_t,
        C_t,
        \widetilde{\bO}_t^{(\Delta)}
    \right),
    \label{eq:event_context_evidence_tuple}
\end{equation}
and therefore depend on the evidence actually available at the decision
node under the usefulness horizon $\Delta$.

\subsection{Information-Theoretic \Gls{eier} and \Gls{eir}}
\label{subsec:information_theoretic_eier}

Before incorporating the admitted evidence, the uncertainty about the event
is $H(E_t\mid C_t)$. Under usefulness horizon $\Delta$, the residual
uncertainty after incorporating the admitted evidence is
\begin{equation}
    H\left(
        E_t\mid C_t,\widetilde{\bO}_t^{(\Delta)}
    \right).
    \label{eq:residual_event_uncertainty_context}
\end{equation}
The information-theoretic \gls{eier} quantifies the fraction of the
context-conditioned event uncertainty that remains after this
update~\cite{CoverThomas2006,gamal2026information}.

\begin{definition}
For $H(E_t\mid C_t)>0$, the information-theoretic \gls{eier} under
usefulness horizon $\Delta$ is
\begin{equation}
    \mathrm{EIER}_{\mathrm{IT}}(\Delta)
    =
    \frac{
        H\left(
            E_t\mid C_t,\widetilde{\bO}_t^{(\Delta)}
        \right)
    }{
        H(E_t\mid C_t)
    }.
    \label{eq:eier_it_context}
\end{equation}
\end{definition}

Since conditioning cannot increase entropy,
$0\leq\mathrm{EIER}_{\mathrm{IT}}(\Delta)\leq1$. A value of zero means that
the admitted evidence resolves all context-conditioned event uncertainty
on average, whereas a value of one means that it provides no additional
information about the event beyond the context.

\par The corresponding uncertainty reduction follows from the conditional
\gls{mi} identity~\cite{CoverThomas2006}
\begin{equation}
\begin{split}
    I\left(
        E_t;
        \widetilde{\bO}_t^{(\Delta)}
        \mid C_t
    \right)
    &=
    H(E_t\mid C_t)\\
    &\quad-
    H\left(
        E_t\mid C_t,\widetilde{\bO}_t^{(\Delta)}
    \right).
    \label{eq:conditional_event_mi_identity}
\end{split}
\end{equation}
Accordingly, \gls{eir} is defined as the fraction of the
context-conditioned uncertainty resolved by the admitted evidence:
\begin{equation}
    \mathrm{EIR}(\Delta)
    =
    \frac{
        I\left(
            E_t;
            \widetilde{\bO}_t^{(\Delta)}
            \mid C_t
        \right)
    }{
        H(E_t\mid C_t)
    }
    =
    1-\mathrm{EIER}_{\mathrm{IT}}(\Delta).
    \label{eq:eir_context}
\end{equation}
Thus, a higher \gls{eir} indicates that the admitted evidence resolves a
larger fraction of the uncertainty remaining after conditioning on context.
Its value therefore depends not only on which cues are admitted, but also
on how informative those cues are about the event.

\begin{remark}
The normalised definitions require $H(E_t\mid C_t)>0$. If
$H(E_t\mid C_t)=0$, the event is already determined by the context almost
surely and there is no remaining uncertainty for the current evidence to
reduce. We therefore treat this degenerate case as resolved by context and
exclude it from normalised \gls{eier} and \gls{eir} evaluation.
\end{remark}

\subsection{Evidence-Limited Event Error}
\label{subsec:evidence_limited_event_error}

The information-theoretic \gls{eier} quantifies residual event uncertainty
rather than the probability of an incorrect event decision. To relate this
uncertainty to decision performance, we consider the minimum event error
supported by the same context and admitted evidence. Under equal costs for
all incorrect event decisions, the Bayes-optimal estimate is
\begin{equation}
    \hat{E}_{t}^{\star,(\Delta)}
    =
    \arg\max_{e\in\cE}
    p\left(
        e\mid C_t,\widetilde{\bO}_t^{(\Delta)}
    \right),
    \label{eq:bayes_event_decoder}
\end{equation}
with corresponding error probability
\begin{equation}
\begin{split}
    P_e^\star(\Delta)
    &=
    \Prb\left(
        \hat{E}_{t}^{\star,(\Delta)}\neq E_t
    \right)\\
    &=
    1-
    \E\left[
        \max_{e\in\cE}
        p\left(
            e\mid C_t,\widetilde{\bO}_t^{(\Delta)}
        \right)
    \right].
    \label{eq:minimum_event_error}
\end{split}
\end{equation}
This is the minimum average event error achievable by any decision rule
using the available context and admitted evidence
~\cite{varshney1997distributed}. It therefore defines the
implementation-independent, evidence-limited benchmark induced by the
decision-node information law~\cite{gamal2026information}. When that law is known or represented by a tractable probabilistic model, the benchmark can be evaluated directly; otherwise, it remains the corresponding Bayes limit against which an implementable inference engine is conceptually measured.
\par The residual entropy in
\eqref{eq:residual_event_uncertainty_context} and
$P_e^\star(\Delta)$ capture different properties of the posterior
distribution. The former quantifies its average uncertainty across the
event alphabet, whereas the latter is determined by its largest posterior
probability. Consequently,
$\mathrm{EIER}_{\mathrm{IT}}(\Delta)$ does not uniquely determine
$P_e^\star(\Delta)$ and does not, in general, induce the same ordering
across different evidence configurations. Residual entropy nevertheless
constrains the minimum achievable event error through the entropy-based
bound developed in the next section~\cite{gamal2026information}.

\subsection{Decoder-Specific Event Error}
\label{subsec:decoder_specific_event_error}

A practical inference engine can differ from the Bayes-optimal rule when
the posterior in \eqref{eq:true_posterior} is not directly available or is
only approximately represented by $q_\theta$. The operational event-error
probability of the resulting estimate in \eqref{eq:map_event_decoder} is
\begin{equation}
    P_{e,\mathrm{op}}(\Delta;\theta)
    =
    \Prb\left(
        \hat{E}_t^{(\Delta)}\neq E_t
    \right).
    \label{eq:decoder_specific_event_error}
\end{equation}
Since $P_e^\star(\Delta)$ is the minimum error achievable from the same
decision-node information,
\begin{equation}
    P_{e,\mathrm{op}}(\Delta;\theta)
    \geq
    P_e^\star(\Delta).
    \label{eq:decoder_bayes_error_relation}
\end{equation}
When the Bayes benchmark is available or can be estimated, the difference
between these quantities measures the additional performance loss
associated with the implemented inference engine. If $q_\theta$ coincides
with the underlying posterior and MAP decoding is used, the inequality
holds with equality and no decoder-specific gap remains.

\par For a labelled evaluation set
\begin{equation}
    \left\{
        \left(
            E_i,
            C_i,
            \widetilde{\bO}_i^{(\Delta)}
        \right)
    \right\}_{i=1}^{N},
    \label{eq:event_labelled_instances}
\end{equation}
the operational event error can be estimated as
\begin{equation}
    \widehat{P}_{e,\mathrm{op}}(\Delta;\theta)
    =
    \frac{1}{N}
    \sum_{i=1}^{N}
    \mathbf{1}
    \left[
        \hat{E}_i^{(\Delta)}\neq E_i
    \right],
    \label{eq:empirical_decoder_specific_event_error}
\end{equation}
where $\mathbf{1}[\cdot]$ is the indicator function and the evaluation set
is independent of the data used to train the inference engine.

\par The distinction therefore separates ambiguity attributable to the available
evidence from additional loss introduced when the implemented inference
engine only approximates the Bayes posterior, owing for example to model
mismatch, finite training data, or computational constraints. The \gls{eier} provides the
complementary uncertainty-based characterisation of the admitted evidence.
From an active-inference perspective, \gls{eier} and \gls{eir} characterise
the uncertainty remaining and resolved, respectively, in the perceptual
belief update.

\section{Event-Inference Limits and Wireless Design}
\label{sec:it_bound_event_error}

We now connect $\mathrm{EIER}_{\mathrm{IT}}$ to evidence-limited event
error and event-aware wireless design. First, we derive an entropy-based
lower bound on the minimum achievable event error; we then examine how
wireless control shapes the evidence admitted within the prescribed
usefulness horizon.

\subsection{Entropy-Based Bound on Event Error}
\label{subsec:entropy_based_event_error_bound}

Recall that $P_e^\star(\Delta)$ is the minimum event error achievable from
the decision-node information
$\left(C_t,\widetilde{\bO}_t^{(\Delta)}\right)$. The following result
relates this evidence-limited error to the residual uncertainty about the
event~\cite{Fano1961,CoverThomas2006,wainwright2019high}.

\begin{proposition}
\label{prop:fano_event_error_bound}
For an event alphabet $\cE$ of size $K\geq2$,
\begin{equation}
    P_e^\star(\Delta)
    \geq
    \left[
        \frac{
            H\left(
                E_t\mid C_t,\widetilde{\bO}_t^{(\Delta)}
            \right)-1
        }{
            \log_2 K
        }
    \right]^+,
    \label{eq:fano_lower_context}
\end{equation}
where $[x]^+=\max\{x,0\}$. When $H(E_t\mid C_t)>0$, this can equivalently
be written as
\begin{equation}
    P_e^\star(\Delta)
    \geq
    \left[
        \frac{
            H(E_t\mid C_t)\,
            \mathrm{EIER}_{\mathrm{IT}}(\Delta)-1
        }{
            \log_2 K
        }
    \right]^+.
    \label{eq:fano_eier_context}
\end{equation}
\end{proposition}

\begin{IEEEproof}
Let
\begin{equation}
    Z_t
    =
    \left(
        C_t,\widetilde{\bO}_t^{(\Delta)}
    \right)
\end{equation}
denote the decision-node information. For any decision rule
$\hat{E}_t=g(Z_t)$ with error probability
$P_e=\Prb(\hat{E}_t\neq E_t)$, Fano's inequality gives
\begin{equation}
    H(E_t\mid\hat{E}_t)
    \leq
    h_b(P_e)+P_e\log_2(K-1),
    \label{eq:fano_standard}
\end{equation}
where $h_b(\cdot)$ is the binary entropy function. Since $\hat{E}_t$ is a
function of $Z_t$,
\begin{equation}
    H(E_t\mid Z_t)
    \leq
    H(E_t\mid\hat{E}_t).
\end{equation}
Using $h_b(P_e)\leq1$ and
$\log_2(K-1)\leq\log_2K$ yields
\begin{equation}
    H(E_t\mid Z_t)
    \leq
    1+P_e\log_2K.
\end{equation}
Hence,
\begin{equation}
    P_e
    \geq
    \left[
        \frac{
            H(E_t\mid Z_t)-1
        }{
            \log_2K
        }
    \right]^+.
\end{equation}
Since this holds for every decision rule based on $Z_t$, it also holds for
the Bayes-optimal rule and therefore for $P_e^\star(\Delta)$.
Substituting
$Z_t=\left(C_t,\widetilde{\bO}_t^{(\Delta)}\right)$ gives
\eqref{eq:fano_lower_context}. Finally, for $H(E_t\mid C_t)>0$,
\begin{equation}
    H\left(
        E_t\mid C_t,\widetilde{\bO}_t^{(\Delta)}
    \right)
    =
    H(E_t\mid C_t)\,
    \mathrm{EIER}_{\mathrm{IT}}(\Delta),
\end{equation}
which gives \eqref{eq:fano_eier_context}.
\end{IEEEproof}

Proposition~\ref{prop:fano_event_error_bound} provides an information-theoretic connection between residual event uncertainty and evidence-limited decision performance. When the residual entropy is sufficiently large, it imposes a strictly positive lower bound on the event error achievable from the same context and admitted evidence. Reducing
$\mathrm{EIER}_{\mathrm{IT}}(\Delta)$ relaxes this constraint, although the exact value of $P_e^\star(\Delta)$ remains determined by the full posterior distribution over $\cE$~\cite{gamal2026information}. Wireless design, by shaping which cues become available within the prescribed usefulness horizon, consequently shapes the admitted evidence and its residual event uncertainty.

\subsection{Event-Aware Wireless Design}
\label{subsec:event_aware_wireless_control}

Wireless control shapes the admitted evidence by influencing whether and
when individual cues become available at the decision node. Let
$\bm{u}\in\mathcal{U}$ denote the available control variables, including,
for example, scheduling, retransmission, resource allocation, path
selection, compression, access priority, and cue-reliability
allocation~\cite{chiariotti2022voi,shao2023task}. For a prescribed
usefulness horizon $\Delta$, these controls induce the evidence process
$\widetilde{\bO}_t^{(\Delta)}(\bm{u})$ and hence
\begin{equation}
    \mathrm{EIER}_{\mathrm{IT}}(\Delta,\bm{u})
    =
    \frac{
        H\left(
            E_t\mid
            C_t,\widetilde{\bO}_t^{(\Delta)}(\bm{u})
        \right)
    }{
        H(E_t\mid C_t)
    }.
    \label{eq:eier_with_control}
\end{equation}
Wireless design can therefore be evaluated according to the residual event
uncertainty induced by the resulting admitted evidence. Alternatively, it
can be evaluated directly through the corresponding Bayes event error
$P_e^\star(\Delta,\bm{u})$. These criteria capture complementary properties
of the event belief: \gls{eier} depends on the full spread of the posterior
distribution, whereas the Bayes event error depends on its largest posterior
probability. This distinction allows
the wireless design objective to be matched to whether the application
prioritises belief quality or event-decision accuracy.

\par More generally, for a prescribed usefulness horizon $\Delta$ and
operating condition $\xi$, let
$\mathcal{J}_E(\Delta,\bm{u};\xi)$ denote the selected event-level
performance criterion, such as
$\mathrm{EIER}_{\mathrm{IT}}(\Delta,\bm{u};\xi)$ or
$P_e^\star(\Delta,\bm{u};\xi)$. Event-aware wireless design can then seek
\begin{equation}
    \bm{u}_E^\star(\Delta,\xi)
    \in
    \arg\min_{\bm{u}\in\mathcal{U}(\Delta,\xi)}
    \mathcal{J}_E(\Delta,\bm{u};\xi),
    \label{eq:general_event_aware_design}
\end{equation}
where $\mathcal{U}(\Delta,\xi)$ denotes the set of feasible wireless
policies under the corresponding temporal, resource, and operating
constraints. Alternatively, when wireless resources are themselves the
design objective, an event-level reliability requirement can be imposed as
\begin{equation}
\begin{aligned}
    \min_{\bm{u}\in\mathcal{U}(\Delta,\xi)}
    \quad
    & \mathcal{C}(\bm{u})\\
    \mathrm{s.t.}
    \quad
    &
    \mathcal{J}_E(\Delta,\bm{u};\xi)
    \leq
    \epsilon_E ,
\end{aligned}
\label{eq:event_aware_resource_design}
\end{equation}
where $\mathcal{C}(\bm{u})$ represents the relevant wireless-resource cost.
These formulations provide a general interface through which wireless
control can be driven by the quality of the resulting event inference, as discussed in the following.

\par One representative specialisation is cue prioritisation under
constrained wireless resources. Using Bayes event error as the event-level
criterion, the event-aware policy becomes
\begin{equation}
    \bm{u}_{E}^{\star}(\Delta,\xi)
    \in
    \arg\min_{\bm{u}\in\mathcal{U}(\Delta,\xi)}
    P_e^\star(\Delta,\bm{u};\xi).
    \label{eq:event_aware_priority_policy}
\end{equation}
The resulting allocation accounts for how the timely admission of each cue
affects event-decision performance, rather than prioritising cues according
to link-level delivery characteristics alone. A second specialisation characterises the minimum usefulness horizon
within which a target event-error level can be supported. For operating
condition $\xi$, define
\begin{equation}
    \Delta_E^\star
    \left(
        \epsilon_E^{(P)};\xi
    \right)
    =
    \inf
    \left\{
        \Delta:
        \min_{\bm{u}\in\mathcal{U}(\Delta,\xi)}
        P_e^\star(\Delta,\bm{u};\xi)
        \leq
        \epsilon_E^{(P)}
    \right\}.
    \label{eq:min_event_useful_twi}
\end{equation}
Here, $\Delta_E^\star$ is not a wireless control variable, but a network
capability measure: it gives the shortest task-imposed usefulness horizon
for which the feasible wireless design can support the required event-error
level. Viewed across operating conditions, it provides a temporal measure
of event-inference coverage.
\par Event-aware wireless design therefore operates at the evidence interface,
shaping the sensory evidence supplied to the perceptual belief update and,
in an active-inference interpretation, the information available for
subsequent belief-driven decision and action. At the protocol level, these
event-aware decisions can be realised through source-specific controls such
as resource allocation, access priority, or retransmission limits. In the
centralised setting considered here, the \gls{bs} determines the corresponding
control configuration and conveys it to the relevant cue sources through
control signalling.
Section~\ref{sec:operationalisation_activity_event} instantiates the two
specialisations above to study cue prioritisation and event-inference
coverage under event-aware and link-level cue-reliability allocation.

\section{\Gls{twi}-Gated Multimodal Activity Inference}
\label{sec:operationalisation_activity_event}
We now instantiate the indoor activity-inference setting introduced as a
guiding example in Section~\ref{sec:introduction}. We consider a
wireless-enabled \gls{pai} system that observes a sensorised indoor
environment through distributed wearable, object-attached, and ambient
sensing sources. The resulting multimodal cues are delivered to the
\gls{bs}, which acts as the edge decision node for activity-event inference.
Continuous sensor streams are partitioned into activity-event instances,
each associated with an event reference time $t$ and a multiclass physical
event $E_t$. For example, the event alphabet can
distinguish among locomotion or body-motion activity, object-interaction
activity, and background or other room-level activity. For each event
instance, the \gls{bs} combines its context-conditioned prior with the
multimodal cues admitted by $t+\Delta$ to form a posterior belief over the
activity event. This provides a concrete three-class instantiation of the
general framework developed above, which extends to other finite event
alphabets and multimodal cue configurations without changing its structure. 

\par This operationalisation couples an activity-event inference layer with a
wireless evidence-admission layer. The activity-inference layer is
instantiated using the OPPORTUNITY dataset~\cite{opportunity_dataset}. Its
activity annotations provide the ground-truth event labels, while wearable,
object-attached, and ambient sensor streams are grouped into
modality-specific cue sources. The corresponding sensor data are conveyed
to the \gls{bs}, where cue-specific supervised-learning classifiers produce
discrete observations $Y_t^{(m)}$~\cite{roggen2010collecting,opportunity_dataset}.
These classifier outputs instantiate the modality-specific cues introduced
in the system model, with empirical class-conditional statistics
\begin{equation}
    \Prb
    \left(
        Y_t^{(m)}=y
        \mid
        E_t=e
    \right)
\end{equation}
providing the cue likelihoods used to characterise their informativeness
and update the context-conditioned prior.

\par The wireless evidence-admission layer is instantiated using Wi3R indoor
channel records from the Wireless Indoor Simulations (WiInSim) dataset by
Qualcomm AI Research~\cite{qualcomm_indoor_dataset,wiinsim2025}. Each
wireless realisation specifies the cue-source-to-\gls{bs} geometry and link
conditions governing delivery of the corresponding sensor data. Together
with the prescribed usefulness horizon $\Delta$, these conditions determine
which cue sources contribute observations to the current event-belief
update. The operating condition $\boldsymbol{\xi}$ represents the relevant
wireless environment, while the resulting admitted cue set determines the
posterior event belief and the associated \gls{eier}, \gls{eir}, and Bayes
event error.
\begin{figure*}[!t]
    \centering
    \includegraphics[width=\textwidth]{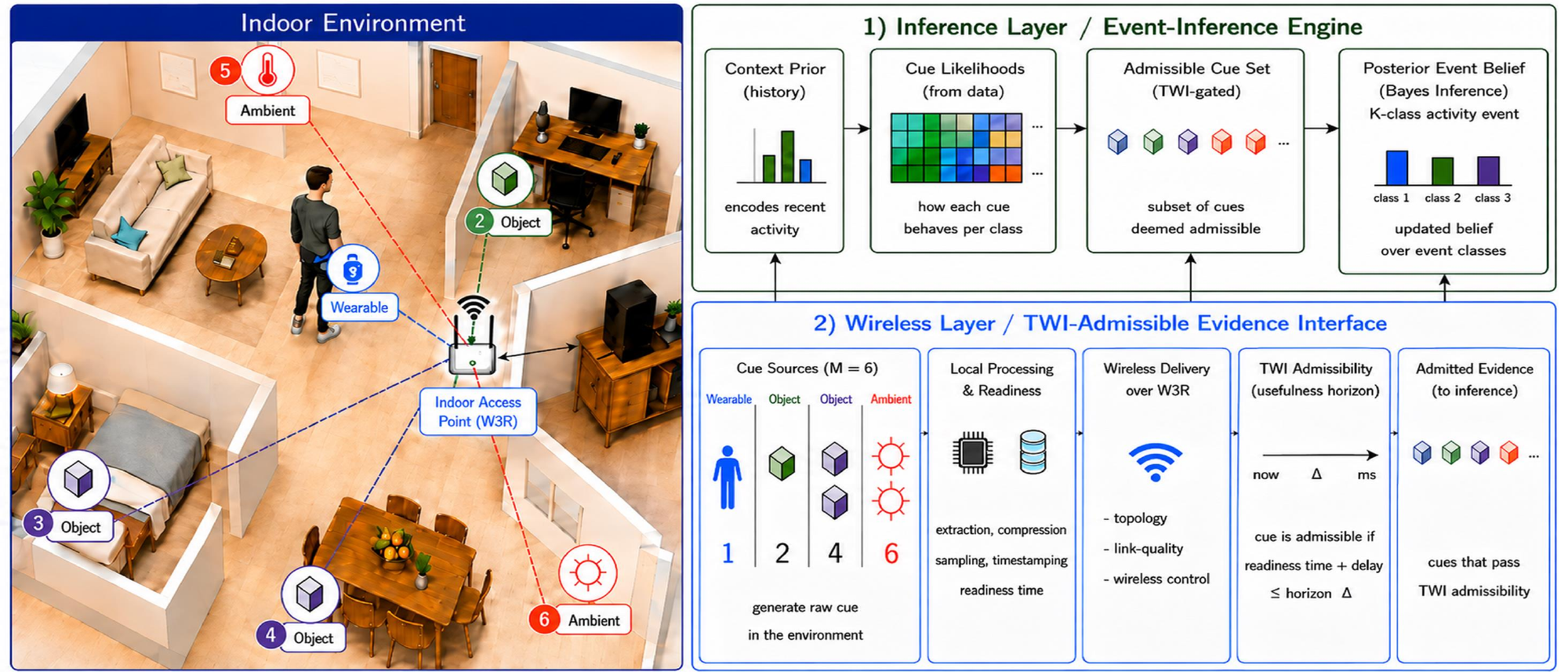}
    \caption{\small Indoor instantiation of the \gls{twi}-gated multimodal
activity-inference framework. Six distributed cue sources generate
multimodal cues whose wireless delivery and \gls{twi}-based admission
determine the evidence available at the \gls{bs}. The inference engine
combines the admitted cue set with a context-conditioned prior and
empirical cue likelihoods to form the posterior belief over three
activity classes.\vspace{-0.5cm}}
    \label{fig:section_v_instantiation}
\end{figure*}

\subsection{Activity State and Cue Observation Model}
\label{subsec:activity_event_cue_model}
With $t$ denoting the event reference time of an activity-inference
instance, let
$E_t\in\mathcal{E}$ denote the corresponding activity event, where
$\mathcal{E}=\{1,\ldots,K\}$ is a finite activity alphabet. In activity
inference, $K$ is fixed by the activity taxonomy adopted for the application,
reflecting the activity distinctions relevant to the task and supported by
the available sensing modalities. In the OPPORTUNITY-based instantiation, the original activity annotations are
therefore coarsened into a small number of interpretable activity classes.
Each labelled activity window defines one inference instance and is assigned
to one of these classes. Alternative activity taxonomies can be accommodated
by changing the event alphabet and the corresponding dimensions of the
cue-likelihood matrices, posterior belief, and entropy calculations.
\par Before incorporating the current cues, the decision node holds the
context-conditioned prior
\begin{equation}
    \lambda_{t,e}
    =
    \Prb(E_t=e\mid C_t),
    \quad e\in\mathcal{E},
    \label{eq:activity_context_prior}
\end{equation}
where $C_t$ contains the information available before the current cue set is
admitted. The vector
\begin{equation}
    \boldsymbol{\lambda}_t
    =
    [\lambda_{t,1},\ldots,\lambda_{t,K}]^{\mathsf T}
    \label{eq:activity_prior_vector}
\end{equation}
therefore represents the prior belief over the activity classes before the
current event-belief update.

\par The distributed sensor streams are grouped into cue sources
$m\in\mathcal{M}$, where $\mathcal{M}=\{1,\ldots,M\}$. For source $m$,
window-level features $\mathbf{X}_t^{(m)}$ are encoded as
\begin{equation}
    \mathbf{Z}_t^{(m)}
    =
    \phi_m\left(\mathbf{X}_t^{(m)}\right),
    \label{eq:compact_evidence_vector}
\end{equation}
and, following delivery, a source-specific classifier $g_m$ at the
\gls{bs} produces
\begin{equation}
    Y_t^{(m)}
    =
    g_m\left(\mathbf{Z}_t^{(m)}\right)
    \in\mathcal{Y}_m.
    \label{eq:multiclass_local_cue}
\end{equation}
Here, $Y_t^{(m)}$ instantiates the decision-node cue $O_t^{(m)}$ of the
general system model, obtained from the delivered representation
$\mathbf{Z}_t^{(m)}$. We take $\mathcal{Y}_m=\mathcal{E}$, so each
source produces an activity-class observation over the common event
alphabet; the final event estimate is obtained only after combining the
admitted cues with the context prior.

\par The informativeness of source $m$ is characterised by the empirical
likelihood matrix
\begin{equation}
    \mathbf{A}^{(m)}_{y,e}
    =
    \Prb\left(
        Y_t^{(m)}=y
        \mid
        E_t=e
    \right),
    \quad
    y\in\mathcal{Y}_m,\;
    e\in\mathcal{E},
    \label{eq:multiclass_cue_likelihood_matrix}
\end{equation}
with
\begin{equation}
    \sum_{y\in\mathcal{Y}_m}
    \mathbf{A}^{(m)}_{y,e}
    =
    1,
    \quad e\in\mathcal{E}.
    \label{eq:multiclass_cue_likelihood_normalisation}
\end{equation}
The matrix therefore captures which activity classes the cue source
distinguishes well and which it tends to confuse.

\par For tractable multimodal fusion, cue observations are assumed
conditionally independent given the activity event. For an admitted subset
$S\subseteq\mathcal{M}$ and realisation
$\boldsymbol{y}_S=(y_m)_{m\in S}$,
\begin{equation}
    L_e(\boldsymbol{y}_S)
    =
    \prod_{m\in S}
    \mathbf{A}^{(m)}_{y_m,e},
    \quad e\in\mathcal{E}.
    \label{eq:multiclass_subset_likelihood}
\end{equation}
For statistically dependent cues, this product can be replaced by an
empirical or learned joint likelihood without changing the subsequent
posterior, entropy, or event-error definitions.
Together, $\boldsymbol{\lambda}_t$ and $L_e(\boldsymbol{y}_S)$ provide the
prior and likelihood terms for the posterior event-belief update.

\subsection{Wireless Evidence Delivery and \Gls{twi} Admission}
\label{subsec:activity_twi_delivery}

We next specify how the representation generated by each cue source becomes
available for the posterior event-belief update. For an activity-inference instance associated with event reference time $t$, source $m$ forms the compact representation
$\mathbf{Z}_t^{(m)}$ and transmits it to the \gls{bs}. If the corresponding
packet is successfully delivered within the usefulness horizon $\Delta$,
the \gls{bs} applies the classifier $g_m$ and admits the resulting cue
$Y_t^{(m)}$ into the posterior update. Otherwise, that cue is unavailable
for the current inference instance. Wireless delivery and \gls{twi}
admission therefore induce a random subset of cue observations contributing
to the posterior belief over $E_t$.
\par Let
$\boldsymbol{\xi}=(\xi_1,\ldots,\xi_M)$ denote the operating condition,
where $\xi_m$ collects the link-level descriptors governing delivery from
source $m$ to the \gls{bs}. The numerical evaluation instantiates
$\boldsymbol{\xi}$ using indoor channel records, while the admission model
requires only the resulting packet-failure and delay distributions under
the selected wireless control $\bm{u}$.
\par Fig.~\ref{fig:section_v_instantiation} illustrates the representative
direct-delivery topology considered in the numerical study. For source $m$,
let $\Gamma_m$ denote the path from the cue-source location
$\mathbf{r}_m$ to the \gls{bs} location $\mathbf{b}$:
\begin{equation}
    \Gamma_m:
    \mathbf{r}_m
    \rightarrow
    \mathbf{b},
    \quad
    m=1,\ldots,M.
    \label{eq:activity_direct_cue_path}
\end{equation}
The same formulation accommodates multi-hop delivery by replacing
$\Gamma_m$ with the corresponding relay path and including its associated
delay components; the present instantiation considers direct delivery.

\par The event-relative end-to-end delay along $\Gamma_m$ is written as
\begin{equation}
    T_{\Gamma_m}
    =
    T_{\mathrm{prop},m}
    +
    T_{\mathrm{comp},m}
    +
    T_{\mathrm{comm},m},
    \label{eq:activity_end_to_end_cue_delay}
\end{equation}
following the \gls{twi} path-delay structure. Here,
$T_{\mathrm{prop},m}$ represents physical event propagation to the sensing
source, $T_{\mathrm{comp},m}$ the cue-readiness time, and
$T_{\mathrm{comm},m}$ the wireless delivery delay. In the considered
indoor setting, physical event-propagation delays are negligible relative
to cue formation and wireless delivery, and hence
\begin{equation}
    T_{\mathrm{prop},m}=0,
    \quad
    m=1,\ldots,M.
    \label{eq:activity_zero_propagation_delay}
\end{equation}
Any sensing or accumulation latency associated with forming the
window-level representation is included in $T_{\mathrm{comp},m}$.

\par Specifically, let the aggregate cue-readiness time be
\begin{equation}
    B_m
    \sim
    \mathrm{Unif}
    \left[
        B_{\min,m},B_{\max,m}
    \right].
    \label{eq:activity_buffering_time}
\end{equation}
This term accounts for cue readiness, construction of the window-level
features, and encoding of $\mathbf{Z}_t^{(m)}$. Because transmission is frame based, the resulting
computation/readiness contribution is quantised to the frame duration
$T_f$~\cite{mishra2025temporal}:
\begin{equation}
    T_{\mathrm{comp},m}
    =
    T_f
    \left\lceil
        \frac{B_m}{T_f}
    \right\rceil .
    \label{eq:activity_frame_quantised_comp_delay}
\end{equation}
The \gls{bs}-side classifier evaluation time is taken as fixed. Accordingly,
when enforcing the task-imposed usefulness horizon $\Delta$, this fixed
inference time is deducted from the time budget available for cue formation
and delivery.

\par The wireless delay $T_{\mathrm{comm},m}$ follows a frame-level
grant-free transmission model with ACK/NACK feedback and at most
$N_m^{\max}$ attempts. Each attempt occupies one frame of duration $T_f$.
If the first successful attempt occurs at index $J_m=j$, then we get
\begin{equation}
    T_{\mathrm{comm},m}
    =
    jT_f .
    \label{eq:activity_comm_delay_frame_model}
\end{equation}
Let $\epsilon_m$ denote the per-attempt packet-error probability under
$\xi_m$ and control $\bm{u}$. In the numerical implementation,
$\epsilon_m$ is obtained from the received power, noise floor, packet
length, and an IEEE~802.15.4-like packet-error
model~\cite{ieee802154_plr}. Under conditionally independent attempts,
\begin{equation}
    p_{\mathrm{GF},m}
    =
    1-\epsilon_m,
    \qquad
    P_{\Gamma_m}^{\mathrm{drop}}
    =
    \epsilon_m^{N_m^{\max}}.
\end{equation}
Conditioned on successful delivery,
\begin{equation}
    \Prb
    \left(
        J_m=j
        \mid
        \mathrm{no\ drop}
    \right)
    =
    \frac{
        \epsilon_m^{j-1}
        \left(1-\epsilon_m\right)
    }{
        1-\epsilon_m^{N_m^{\max}}
    },
    \quad
    j=1,\ldots,N_m^{\max}.
    \label{eq:gf_success_attempt_distribution}
\end{equation}

\par Representing packet failure by $T_{\Gamma_m}=\infty$, the admission
indicator for cue $m$ is
\begin{equation}
    A_{t,m}^{(\Delta)}
    =
    \mathbf{1}
    \left\{
        T_{\Gamma_m}\leq\Delta
    \right\},
    \label{eq:activity_admissibility_indicator}
\end{equation}
where $A_{t,m}^{(\Delta)}=1$ indicates successful delivery within the
usefulness horizon. As in the general system model, simultaneity and
causal-order consistency are assumed to be handled upstream; they can
instead be incorporated as additional admission conditions without changing
the subsequent event-inference formulation.

\par The corresponding admission probability is
\begin{align}
    r_m(\Delta;\xi_m,\bm{u})
    &=
    \Prb
    \left(
        T_{\Gamma_m}\leq\Delta
        \mid
        \xi_m,\bm{u}
    \right)
    \label{eq:activity_cue_admissibility_probability}
    \\
    &=
    \left(
        1-P_{\Gamma_m}^{\mathrm{drop}}
    \right)
    F_{\Gamma_m}
    \left(
        \Delta
        \mid
        \mathrm{no\ drop},\xi_m,\bm{u}
    \right),
    \nonumber
\end{align}
where $F_{\Gamma_m}$ is the conditional \gls{cdf} of the end-to-end delay
given successful delivery. The admitted cue-source subset is
\begin{equation}
    S_\Delta(t)
    =
    \left\{
        m\in\mathcal{M}:
        A_{t,m}^{(\Delta)}=1
    \right\}.
    \label{eq:activity_admissible_subset}
\end{equation}

\par Under conditionally independent direct deliveries, the probability of
admitting subset $S\subseteq\mathcal{M}$ is
\begin{align}
    \rho_\Delta(S;\boldsymbol{\xi},\bm{u})
    &=
    \Prb
    \left(
        S_\Delta(t)=S
        \mid
        \boldsymbol{\xi},\bm{u}
    \right)
    \label{eq:activity_independent_subset_distribution}
    \\
    &=
    \prod_{m\in S}
    r_m(\Delta;\xi_m,\bm{u})
    \prod_{j\in\mathcal{M}\setminus S}
    \left[
        1-r_j(\Delta;\xi_j,\bm{u})
    \right].
    \nonumber
\end{align}
For coupled deliveries, this product is replaced by the corresponding joint
distribution of
$\left(A_{t,1}^{(\Delta)},\ldots,A_{t,M}^{(\Delta)}\right)$.

\par The decision-node evidence can therefore be represented as
\begin{equation}
    \widetilde{\bO}_t^{(\Delta)}
    \equiv
    \left(
        S_\Delta(t),
        \left\{
            Y_t^{(m)}:
            m\in S_\Delta(t)
        \right\}
    \right).
    \label{eq:activity_twi_admissible_evidence}
\end{equation}
\par Admission is assumed conditionally independent of the activity event
and cue realisations given $\boldsymbol{\xi}$ and $\bm{u}$, so
$S_\Delta(t)$ determines which cue-likelihood factors enter the posterior
without itself providing additional event information. If this assumption
does not hold, the admission mechanism must also be included in the
posterior model and is left for future work.

\subsection{Bayesian Activity Inference and Residual Uncertainty}
\label{subsec:posterior_activity_event}

For an activity-inference instance $t$, let
$S\subseteq\mathcal{M}$ denote the admitted cue subset and
$\boldsymbol{y}_S=(y_m)_{m\in S}$ its realisation. The \gls{bs} combines
the context prior $\boldsymbol{\lambda}_t$ with the likelihoods of the
admitted cues to form the posterior belief
$\boldsymbol{\eta}_t(\boldsymbol{y}_S)
=
\left(\eta_{t,e}(\boldsymbol{y}_S)\right)_{e\in\mathcal{E}}$.
Under the assumed admission model, which is independent of the activity
event and cue realisation conditioned on $\boldsymbol{\xi}$ and $\bm{u}$,
the admitted subset determines which likelihood factors enter the update.
Suppressing conditioning on the fixed operating condition and wireless
control, the posterior probability of activity class $e$ is
\begin{align}
    \eta_{t,e}(\boldsymbol{y}_S)
    &=
    \Prb\left(
        E_t=e
        \mid
        C_t,
        \boldsymbol{Y}_S=\boldsymbol{y}_S,
        S_\Delta(t)=S
    \right)
    \nonumber\\
    &=
    \frac{
        \lambda_{t,e}L_e(\boldsymbol{y}_S)
    }{
        \displaystyle
        \sum_{e'\in\mathcal{E}}
        \lambda_{t,e'}L_{e'}(\boldsymbol{y}_S)
    }.
    \label{eq:activity_posterior}
\end{align}
For $S=\emptyset$, we set $L_e(\emptyset)=1$, so that the posterior reduces
to the context prior $\boldsymbol{\lambda}_t$.
\par The predictive probability of cue realisation $\boldsymbol{y}_S$ is
\begin{align}
    p_{t,S}(\boldsymbol{y}_S)
    &=
    \Prb\left(
        \boldsymbol{Y}_S=\boldsymbol{y}_S
        \mid
        C_t,S_\Delta(t)=S
    \right)
    \nonumber\\
    &=
    \sum_{e\in\mathcal{E}}
    \lambda_{t,e}L_e(\boldsymbol{y}_S),
    \label{eq:predictive_cue_distribution}
\end{align}
where
$\boldsymbol{y}_S\in
\mathcal{Y}_S=\prod_{m\in S}\mathcal{Y}_m$.
The expected posterior entropy conditioned on admitting subset $S$ is
\begin{equation}
    H_{t,S}
    =
    \sum_{\boldsymbol{y}_S\in\mathcal{Y}_S}
    p_{t,S}(\boldsymbol{y}_S)
    H\left(
        \boldsymbol{\eta}_t(\boldsymbol{y}_S)
    \right),
    \label{eq:subset_residual_entropy_activity}
\end{equation}
with
\begin{equation}
    H\left(
        \boldsymbol{\eta}_t(\boldsymbol{y}_S)
    \right)
    =
    -
    \sum_{e\in\mathcal{E}}
    \eta_{t,e}(\boldsymbol{y}_S)
    \log_2
    \eta_{t,e}(\boldsymbol{y}_S).
    \label{eq:posterior_entropy_activity}
\end{equation}
\par Averaging over the random admitted subset gives the residual uncertainty
for instance $t$:
\begin{equation}
    h_t^{(\Delta)}
    (\boldsymbol{\xi},\bm{u})
    =
    \sum_{S\subseteq\mathcal{M}}
    \rho_\Delta
    (S;\boldsymbol{\xi},\bm{u})
    H_{t,S}.
    \label{eq:instance_residual_entropy_activity}
\end{equation}
The corresponding uncertainty before incorporating the current cues is
\begin{equation}
    h_t^{\mathrm{prior}}
    =
    -
    \sum_{e\in\mathcal{E}}
    \lambda_{t,e}\log_2\lambda_{t,e}.
    \label{eq:context_entropy_activity}
\end{equation}
\par The same posterior determines the evidence-limited event error. For an admitted subset $S$, the expected Bayes error is
\begin{equation}
    P_{e,t}^{\star}(S)
    =
    \sum_{\boldsymbol{y}_S\in\mathcal{Y}_S}
    p_{t,S}(\boldsymbol{y}_S)
    \left[
        1-
        \max_{e\in\mathcal{E}}
        \eta_{t,e}(\boldsymbol{y}_S)
    \right],
    \label{eq:activity_subset_bayes_error}
\end{equation}
and averaging over cue admission gives
\begin{equation}
    P_{e,t}^{\star}
    (\Delta;\boldsymbol{\xi},\bm{u})
    =
    \sum_{S\subseteq\mathcal{M}}
    \rho_\Delta
    (S;\boldsymbol{\xi},\bm{u})
    P_{e,t}^{\star}(S).
    \label{eq:activity_bayes_error_instance}
\end{equation}
\par Equations~\eqref{eq:activity_posterior}--%
\eqref{eq:activity_bayes_error_instance} instantiate the inference chain
from the context prior and admitted cue likelihoods to the posterior event
belief, residual uncertainty, and Bayes event error. Here,
$\boldsymbol{\eta}_t$ is the Bayes posterior under the empirical activity
model, so $P_{e,t}^{\star}$ is the corresponding model-based
evidence-limited benchmark. Averaging the residual and prior entropies over the labelled instances
yields the empirical \gls{eier} and \gls{eir}, while averaging the
corresponding Bayes errors gives the empirical evidence-limited event error.

\subsection{Inference and Evaluation Setup}
\label{subsec:activity_evaluation_setup}

\begin{figure}[!t]
    \centering
    \includegraphics[width=0.95\linewidth]{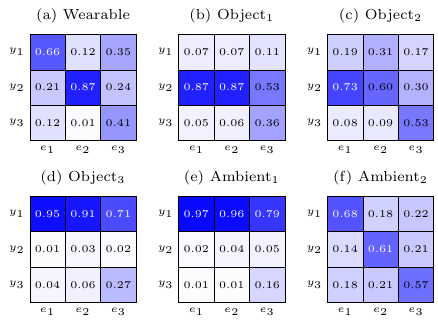}
    \caption{\small Cue-likelihood matrices $\mathbf{A}^{(m)}$ for the
    six cue sources. Colour intensity increases linearly from white~(0) to
    dark blue~(1).\vspace{-0.4cm}}
    \label{fig:cue_likelihood_matrices}
\end{figure}

The OPPORTUNITY recordings are organised by human subject, denoted S1--S4
in our split. Subjects S1--S3 are used to estimate the activity-transition
matrix and train the cue-specific classifiers, while S4 is used to estimate
the empirical cue-likelihood matrices and evaluate the resulting
\gls{twi}-gated inference. The sensor streams are segmented into
non-overlapping windows of $W=60$ consecutive samples, with each window
assigned its majority coarsened activity label. The end of each window is taken as the event reference time $t$ for the corresponding inference instance. At this window length, at least half of the windows contain a single coarsened activity class throughout, supporting their treatment as activity-inference instances.
After removing the first window of each S4 sequence, for which no previous within-sequence activity is available, the evaluation set contains $N=2446$ instances.

\par The context prior is instantiated using the preceding activity state.
The transition matrix estimated from S1--S3 and held fixed during evaluation,
with $t^{-}$ denoting the event reference time of the immediately preceding
activity-inference instance and
$T_{ij}=\Prb(E_t=j\mid E_{t^{-}}=i)$, is
\[
\mathbf{T}
=
\begin{bmatrix}
0.8753 & 0.0440 & 0.0807\\
0.1041 & 0.5599 & 0.3361\\
0.0814 & 0.0866 & 0.8320
\end{bmatrix}.
\]
\begin{table}[!b]
\vspace{-0.2cm}
\centering
\caption{\small Cue-source-specific wireless evaluation parameters.}
\label{tab:cue_source_parameters}
\small
\renewcommand{\arraystretch}{1.08}
\setlength{\tabcolsep}{2pt}
\begin{tabularx}{\columnwidth}
{@{}>{\raggedright\arraybackslash}X
*{5}{>{\centering\arraybackslash}X}@{}}
\hline
\textbf{Source} &
\textbf{Radius} &
\textbf{Bytes} &
$\boldsymbol{B_{\min}}$ &
$\boldsymbol{B_{\max}}$ &
$\boldsymbol{N^{\max}}$ \\
& (m) & & (ms) & (ms) & \\
\hline
Wearable
    & 0.5 & 32 & $5$--$15$  & $25$--$45$  & $2$--$3$ \\
Object$_1$
    & 3   & 16 & $8$--$20$  & $40$--$70$  & $2$--$4$ \\
Object$_2$
    & 3   & 16 & $8$--$20$  & $40$--$70$  & $2$--$4$ \\
Object$_3$
    & 3   & 16 & $8$--$20$  & $45$--$80$  & $2$--$4$ \\
Ambient$_1$
    & 6   & 8  & $10$--$25$ & $60$--$100$ & $3$--$5$ \\
Ambient$_2$
    & 6   & 8  & $10$--$25$ & $60$--$110$ & $3$--$5$ \\
\hline
\end{tabularx}
\end{table}

For each S4 instance, let $c_t$ denote the preceding coarsened activity
label in the same sequence. To avoid assuming perfect context knowledge,
a confusion probability $\nu$ is introduced: probability $1-\nu$ is
assigned to $c_t$, while $\nu$ is distributed uniformly over the remaining
$K-1$ classes. Propagating this belief through $\mathbf{T}$ gives the
context-conditioned prior $\boldsymbol{\lambda}_t^{(\nu)}$. We consider
$\nu=0.15$ and $\nu=0.40$ to represent relatively reliable and degraded
recent-context information, respectively.

\par The inference layer uses $M=6$ cue sources: wearable, object$_1$,
object$_2$, object$_3$, ambient$_1$, and ambient$_2$. Each source is
represented by window-level mean and standard-deviation features and mapped
to a discrete cue using a one-versus-all multiclass logistic classifier
with ridge regularisation. The compact source representations contain
$32$ components for the wearable source, $16$ for each object-oriented
source, and $8$ for each ambient source. Under 8-bit quantisation, these
correspond to payloads of $32$, $16$, and $8$ bytes, respectively.
Fig.~\ref{fig:cue_likelihood_matrices} shows the resulting empirical
cue-likelihood matrices. The wearable source is relatively discriminative
across the activity classes, whereas sources such as ambient$_1$ confuse
several classes and therefore provide weaker activity information.

\par The wireless layer is instantiated through a virtual indoor overlay
using Wi3R received-power records. Since OPPORTUNITY and Wi3R are not
co-recorded, cue-source identities are assigned independently of wireless
link quality. We generate $D=1000$ wireless realisations, with the
source-placement radii and source-specific readiness and retransmission
parameters summarised in Table~\ref{tab:cue_source_parameters}. For each
 realisation $d$ and source $m$, the received power and the
parameters $B_{\min,d,m}$, $B_{\max,d,m}$, and $N_{d,m}^{\max}$ are fixed.
The latter are sampled once from the corresponding ranges in
Table~\ref{tab:cue_source_parameters}, while the readiness time itself is
resampled for each replayed activity-inference instance.

\par The grant-free delivery model uses frames of duration $T_f=10$~ms,
a packet length of $L_{\mathrm{pkt}}=133$~bytes, and a receiver noise floor
of $-100$~dBm. The packet length conservatively accounts for the compact
source representation and protocol overhead. The effective packet-error
probability is obtained by averaging the IEEE~802.15.4-like packet-error
curve over a $5$~dB residual link-margin perturbation using
deterministic samples. For each activity instance and wireless
realisation, the resulting delivery and readiness processes determine the
admitted cue subset, whose likelihood factors are then used in the
posterior activity-event update.

\subsection{\Gls{eir} and Event Error versus Usefulness Horizon}
\label{subsec:activity_window_results}

\begin{figure}[!t]
    \centering
    \includegraphics[width=0.95\linewidth]{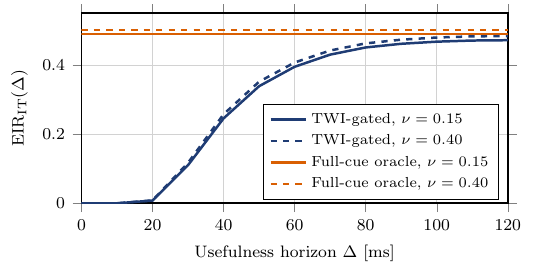}

    \vspace{0.25em}

    \includegraphics[width=0.95\linewidth]{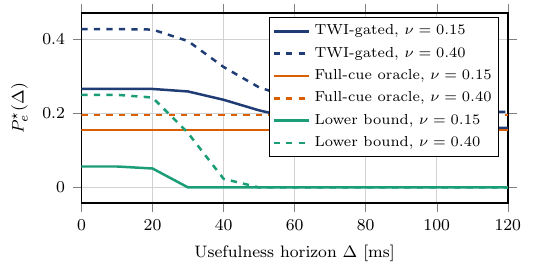}

    \caption{\small \gls{eir}, Bayes event error, and the corresponding
    entropy-based lower bound versus usefulness horizon $\Delta$.\vspace{-0.5cm}}
    \label{fig:activity_eir_vs_delta}
\end{figure}

We first evaluate how the usefulness horizon $\Delta$ affects
\gls{eir}. For each $\Delta$, the distribution of
admitted cue subsets is combined with the context prior and empirical cue
likelihoods to compute $\mathrm{EIR}(\Delta)$ and
$P_e^\star(\Delta)$. As shown in
Fig.~\ref{fig:activity_eir_vs_delta}, increasing $\Delta$ raises the
probability that additional cues become available in time, thereby reducing
residual event uncertainty and Bayes event error. The largest changes occur
over the range of $\Delta$ containing most of the cue-readiness and
retransmission-delay probability mass, after which the gains progressively
saturate.

\par The two context-noise settings illustrate how the value of admitted
evidence depends on the prior information already available at the decision
node. A smaller $\nu$ produces a more informative previous-activity context,
whereas a larger $\nu$ leaves greater prior uncertainty to be resolved by
the incoming cues. As $\Delta$ increases, both the \gls{eir} and Bayes event
error approach their corresponding full-cue benchmarks. Any remaining gap
is attributable to the wireless admission process, including packet drops
that cannot be recovered simply by extending the usefulness horizon $\Delta$.

\par The entropy-based lower bound follows the reduction in residual
uncertainty but is not expected to coincide with the Bayes event error. The
bound depends only on the residual entropy and $K$, whereas
the Bayes error depends on the full posterior distribution. It therefore
serves as an information-theoretic limit on evidence-limited event
performance rather than as a prediction of the achieved Bayes error.

\subsection{Event-Aware Cue Prioritisation}
\label{subsec:event_aware_priority}

We next examine the benefit of incorporating event-inference value into cue
prioritisation rather than relying on link-level timeliness alone. Unlike
the previous subsection, where every \gls{twi}-admissible cue can contribute
to the posterior update, we impose a priority budget $M_{\mathrm p}$ over
the $M=6$ cue sources. Fig.~\ref{fig:priority_gain_heatmap} considers
$M_{\mathrm p}=3$, while Fig.~\ref{fig:priority_curve} varies the budget to
show its effect on Bayes event error.

\par Wireless drops are grouped by the \gls{bs}--activity distance $d$, and
$\Delta$ is evaluated on a finer grid. To isolate the effect of
prioritisation, the retransmission cap is fixed at $N^{\max}=3$ for all
sources and drops. For each drop, the link-level rule selects the sources
with the largest \gls{twi}-admissibility probabilities, whereas the
event-aware rule selects the subset that minimises Bayes event error. The
resulting gain is computed per drop and averaged within each distance bin.

\par Fig.~\ref{fig:priority_gain_heatmap} shows the error reduction
when three of the six cue sources are prioritised. The largest gains occur
where $\Delta$ gives several cues appreciable admission probability but the
priority constraint still requires a choice among them. In this regime,
link-level prioritisation can favour timely but weakly informative or
redundant cues, whereas the event-aware rule favours combinations that
produce a more discriminative posterior.

\begin{figure}[!t]
    \centering
    \includegraphics[width=\columnwidth]{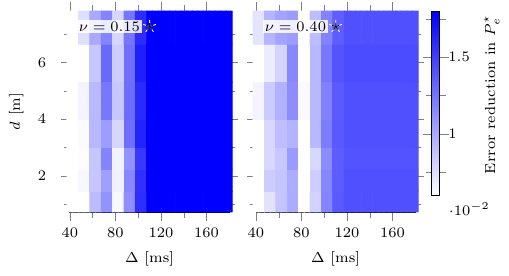}
    \caption{\small Reduction in Bayes event error from event-aware cue
    prioritisation with $M_{\mathrm p}=3$ over the
    distance--$\Delta$ plane. Darker regions indicate larger reductions
    relative to link-level prioritisation.\vspace{-0.4cm}}
    \label{fig:priority_gain_heatmap}
\end{figure}

\begin{figure}[!b]
\vspace{-0.2cm}
    \centering
    \includegraphics[width=0.95\linewidth]{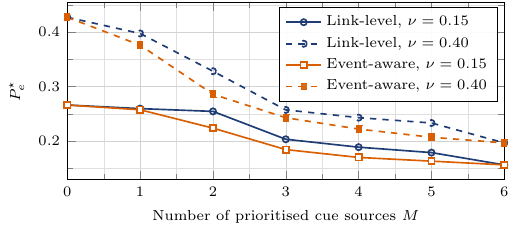}
    \caption{\small Bayes event error versus priority budget
    $M_{\mathrm p}$ at the highlighted distance--usefulness-horizon bin.}
    \label{fig:priority_curve}
\end{figure}

\par The gain also depends on $\nu$. With more reliable
previous-activity context, part of the event uncertainty is already
resolved, making cue value more sensitive to how well a subset
distinguishes the remaining plausible classes. With degraded context, a
broader range of cues can contribute useful evidence. Event-aware
prioritisation remains beneficial, but the distance--$\Delta$ pattern of
the gain changes because cue value is evaluated relative to a different
context-conditioned prior.

\par Fig.~\ref{fig:priority_curve} shows the effect of varying
$M_{\mathrm p}$ at the highlighted operating point. The two rules coincide
when no source or all sources are prioritised, since no subset-selection
choice remains. For intermediate budgets, event-aware prioritisation
achieves lower Bayes event error by selecting cues according to their
contribution to the event decision rather than their admissibility
probability alone.
\subsection{Event-Inference Coverage under a Cue-Reliability Budget}
\label{subsec:activity_coverage_results}

We evaluate event-inference coverage through the minimum $\Delta$
required to support a target Bayes event error under a given spatial
condition and cue-reliability budget. Fig.~\ref{fig:activity_coverage_profiles}
compares this horizon for link-level and event-aware reliability allocation
in the selected distance bin. The horizontal axis gives the network-level
unavailability budget $\epsilon_{\mathrm{net}}$, with both allocation rules
evaluated under the two context-noise levels.

\par For each $\epsilon_{\mathrm{net}}$, the link-level allocation determines
the minimum supportable horizon subject to
$\rho_m\leq r_m(\Delta)$ and
$\sum_m(1-\rho_m)\leq\epsilon_{\mathrm{net}}$.
This yields an event-blind cue-reliability allocation
$\boldsymbol{\rho}_{\mathrm{LL}}^\star$ and its corresponding Bayes event
error. The event-aware allocation uses the same unavailability budget but
redistributes cue reliability to attain the same Bayes event-error target
with the minimum required horizon. The reported values are averaged over
the readiness-time realisations used to estimate $r_m(\Delta)$ and then
over the wireless drops within the selected distance bin.

\begin{figure}[!t]
    \centering
    \includegraphics[width=0.95\columnwidth]{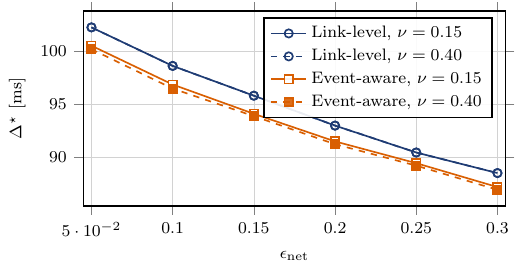}
    \caption{\small Minimum $\Delta$ at the
    $d\approx 8$~m distance bin versus network-level unavailability budget
    under link-level and event-aware cue-reliability allocation.\vspace{-0.5cm}}
    \label{fig:activity_coverage_profiles}
\end{figure}

\par As $\epsilon_{\mathrm{net}}$ increases, the minimum required horizon
decreases because greater aggregate cue unavailability is permitted. For
both context-noise levels, event-aware allocation attains the matched Bayes
event-error target with a shorter horizon than link-level allocation by
directing reliability towards cues with greater inferential value. The gain
is larger for $\nu=0.40$, where the degraded context leaves more event
uncertainty for the admitted cues to resolve; with $\nu=0.15$, the stronger
prior reduces this dependence on the cue set. Thus, event-inference coverage
depends not only on the aggregate reliability available to the network, but
also on how that reliability is distributed across cues with different
inferential value.
\section{Conclusion}
\label{sec:conclusion}
This work developed an \gls{eir} framework for wireless-enabled \gls{pai}
that characterises reliability through the uncertainty resolved about a
task-relevant physical event. The framework distinguishes cue
informativeness, availability, and \gls{twi}-based temporal admissibility,
and defines \gls{eier} and \gls{eir} from the context-conditioned event
uncertainty before and after cue admission. By relating residual uncertainty
to Bayes event error, it further distinguishes evidence-limited from
decoder-specific performance and provides an entropy-based lower bound on
achievable event error. The multimodal activity-inference instantiation
demonstrated how wireless delivery and the usefulness horizon shape the
evidence available for inference, while event-aware cue prioritisation and
reliability allocation improved upon link-level designs by directing
wireless resources towards cues with greater inferential value and
supporting matched event-inference performance at shorter usefulness
horizons. 
\bibliographystyle{ieeetr}
\bibliography{reference}
\end{document}